%% file: main.tex
\documentclass[a4paper]{article}
  \usepackage{fullpage}
  \usepackage{authblk}
  \usepackage[utf8]{inputenc}
  \usepackage{amsmath,amsfonts,amsthm,amssymb}
  \usepackage{thmtools}
  \usepackage{hyperref}
  \usepackage{xspace}
  \usepackage{color}
  \usepackage[shortlabels]{enumitem}
  \usepackage{tikz}
  \usepackage[capitalise]{cleveref}
  \usepackage{tikz-qtree}
  \usetikzlibrary{decorations.pathmorphing}
  \usetikzlibrary{decorations.markings}
  \usetikzlibrary{shapes.geometric, snakes}

  \usepackage{mathtools}
\DeclarePairedDelimiter{\floor}{\lfloor}{\rfloor}
\DeclarePairedDelimiter{\ceil}{\lceil}{\rceil}

\newtheorem{theorem}{Theorem}[section]
\newtheorem{fact}[theorem]{Fact}
\newtheorem{lemma}[theorem]{Lemma}
\newtheorem{observation}[theorem]{Observation}

\newtheorem{claim}[theorem]{Claim}

\theoremstyle{remark}

\newtheorem{remark}[theorem]{Remark}
\newtheorem{example}[theorem]{Example}
\usepackage{environ}
\newcommand{\repeatlemma}[1]{%
	\begingroup
	\renewcommand{\thelemma}{\ref{#1}}%
	\expandafter\expandafter\expandafter\lemma
	\csname replemma@#1\endcsname
	\endlemma
	\endgroup
}

\NewEnviron{replemma}[1]{%
	\global\expandafter\xdef\csname replemma@#1\endcsname{%
		\unexpanded\expandafter{\BODY}%
	}%
	\expandafter\lemma\BODY\unskip\label{#1}\endlemma
}

\newcommand{\cO}{\mathcal{O}}
\newcommand{\Oh}{\mathcal{O}}
\newcommand{\cOtilde}{\tilde{\cO}}
\newcommand{\Ohtilde}{\tilde{\cO}}

\newcommand{\PS}{\mathcal{H}}

\def\dd{\mathinner{.\,.}}
\newcommand{\D}{\mathcal{D}}
\newcommand{\T}{\mathcal{T}}
\renewcommand{\S}{\mathcal{S}}
\newcommand{\B}{\mathcal{B}}

\newcommand{\Occ}{\mathsf{Occ}}
\newcommand{\run}{\mathsf{run}}
\newcommand{\Pref}{\mathsf{Pref}}

\newcommand{\Lab}{\mathcal{L}}

\newcommand{\G}{\mathcal{G}}
\newcommand{\Dec}{\textsc{Exists}}
\newcommand{\R}{\textsc{Report}}
\newcommand{\RD}{\textsc{ReportDistinct}}
\newcommand{\C}{\textsc{Count}}
\newcommand{\CD}{\textsc{CountDistinct}}

\newcommand{\per}{\textsf{per}}
\renewcommand{\exp}{\textsf{exp}}
\newcommand{\occ}{\textsf{output}}

\newcommand{\BRLCP}{\mathit{BoundedLCP}\xspace}

\newcommand{\RunSquares}{\mathit{RunSquares}}
\newcommand{\Sql}{\mathit{BSq}}

\newcommand{\eps}{\delta}

 \newcommand{\defDSproblem}[3]{
  \vspace{2mm}
\noindent\fbox{
  \begin{minipage}{0.96\textwidth}
  \textsc{#1}\\
  {\bf{Input:}} #2  \\
  {\bf{Query:}} #3
  \end{minipage}
  }
  \vspace{2mm}
}

 \usetikzlibrary{shapes,calc,arrows,positioning,decorations.markings,trees}
  \tikzset{
    treenode/.style = {align=center, inner sep=2pt, text centered,
      font=\sffamily},
    arn_r/.style = {treenode, circle, black, fill=black,font=\sffamily\bfseries, draw=black,
      text width=0.2em},
      arn_p/.style = {treenode, circle, line width=0.7mm, blue, fill=black, draw=blue, text width=0.4em},
      arn_p1/.style = {treenode, circle, line width=0.7mm, blue, fill=blue, draw=blue, text width=0.4em},
      arn_pat/.style = {treenode, circle , blue, fill=blue, draw=blue, text width=0.3em},
      arn_t/.style = {treenode, circle, black, thick, double, font=\sffamily\bfseries, draw=black,
      text width=0.2em},
    every edge/.append style={anchor=south,auto=falseanchor=south,auto=false,font=3.5 em},
  }
  \tikzset{
      pattern/.style={postaction={decorate},
          decoration={markings,mark=at position .55 with {\draw[thin,fill, blue] circle (5pt);}}}
  }

\begin{document}
\title{Counting~Distinct~Patterns in~Internal~Dictionary~Matching}
\author[1,2]{Panagiotis Charalampopoulos}
\author[3]{Tomasz Kociumaka}
\author[1]{Manal Mohamed}
\author[2]{Jakub Radoszewski}
\author[2]{Wojciech Rytter}
\author[2]{Juliusz Straszy\'nski}
\author[2]{Tomasz Wale\'n}
\author[2]{Wiktor Zuba}

\affil[1]{
  Department of Informatics, King's College London, UK\protect\\
  \texttt{\{panagiotis.charalampopoulos, manal.mohamed\}@kcl.ac.uk}}
\affil[2]{
	Institute of Informatics, University of Warsaw, Poland\\
	\texttt{\{jrad,rytter,jks,walen,w.zuba\}@mimuw.edu.pl}}
\affil[3]{Department of Computer Science, Bar-Ilan University, Ramat Gan, Israel\\
    \texttt{kociumaka@mimuw.edu.pl}}

\date{\vspace{-.5cm}}
	\maketitle              

\begin{abstract}
We consider the problem of preprocessing a text $T$ of length $n$ and a dictionary $\D$ in order to be able to efficiently answer queries $\CD(i,j)$, that is, given $i$ and $j$ return the number of patterns from $\D$ that occur in the \emph{fragment} $T[i \dd j]$.
The dictionary is \emph{internal} in the sense that each pattern in $\D$ is given as a fragment of $T$.
This way, the dictionary takes space proportional to the number of patterns $d=|\D|$ rather than their total length, which could be $\Theta(n\cdot d)$.
An $\cOtilde(n+d)$-size \footnote{The $\tilde{\cO}(\cdot)$ notation suppresses $\log^{\cO(1)} n$ factors for inputs of size $n$.} data structure that answers $\CD(i,j)$ queries $\cO(\log n)$-approximately in $\cOtilde(1)$ time was recently proposed in a work that introduced internal dictionary matching [ISAAC 2019].
Here we present an $\cOtilde(n+d)$-size data structure that answers $\CD(i,j)$ queries $2$-approximately in $\cOtilde(1)$ time.
Using range queries, for any $m$, we give an $\cOtilde(\min(nd/m,n^2/m^2)+d)$-size data structure that answers $\CD(i,j)$ queries exactly in $\cOtilde(m)$ time.
We also consider the special case when the dictionary consists of all square factors
of the string.
We design an $\Oh(n \log^2 n)$-size data structure that allows us to count distinct squares in a text fragment $T[i \dd j]$ in $\Oh(\log n)$ time.
\end{abstract}

\section{Introduction}
Internal Dictionary Matching was recently introduced in~\cite{DBLP:conf/isaac/Charalampopoulos19} as a generalization of Internal Pattern Matching. In the classical Dictionary Matching problem, we are given a dictionary $\D$ consisting of $d$ patterns, and the goal is to preprocess $\D$ so that, presented with a text $T$, we can efficiently compute the occurrences of the patterns from $\D$ in $T$. In Internal Dictionary Matching, the text $T$ is given in advance, the dictionary $\D$ is a set of fragments of $T$, and the Dictionary Matching queries can be asked for any fragment of $T$.

The Internal Pattern Matching problem consists in preprocessing a text $T$ of length $n$ so that we can efficiently compute the occurrences of a fragment of $T$ in another fragment of~$T$.
A data structure of nearly linear size that allows for sublogarithmic-time Internal Pattern Matching queries was presented in~\cite{DBLP:journals/tcs/KellerKFL14}, while a linear-size data structure allowing for constant-time Internal Pattern Matching queries in the case that the ratio between the lengths of the two factors is constant was presented in~\cite{DBLP:conf/soda/KociumakaRRW15}.
Other types of internal queries have also been studied; we refer the interested reader to~\cite{tomeksthesis}.

In~\cite{DBLP:conf/isaac/Charalampopoulos19}, several types of Internal Dictionary Matching queries 
about  fragments $T[i\dd j]$ in  a string $T$  were considered:
$\Dec(i,j)$, $\R(i,j)$, $\RD(i,j)$, $\C(i,j)$ and $\CD(i,j)$.
Data structures of size $\cOtilde(n+d)$ and query time $\cOtilde(1+\occ)$ were shown for answering each of the first four queries, with $\C$ queries requiring most advanced techniques.
For $\CD$ queries, only a data structure answering these queries $\cO(\log n)$-approximately was shown. 
In this work, we focus on 
more efficient data structures for such queries. 
$\CD$ queries are formally defined as follows.

\defDSproblem{CountDistinct}{A text $T$ of length $n$ and a dictionary $\D$ consisting of $d$ patterns, each given as a fragment $T[a \dd b]$ of $T$
(represented only by integers $a,b$).}
{$\CD(i,j)$: Count all distinct patterns $P \in \D$ that occur in $T[i \dd j]$. }

\noindent Observe that the input size is $n+d$, while the total length of strings in $\D$ could be $\Theta(n\cdot d)$.

We also consider a special case of this problem when the dictionary $\D$ is the set of all squares (i.e., strings of the form $UU$) in $T$. The case that $\D$ is the set of palindromes in $T$ was considered by Rubinchik and Shur in~\cite{DBLP:conf/spire/RubinchikS17}.

\begin{example}\label{ex:1}
  Let us consider the following text:
  \begin{center}
      \begin{tabular}{c|c|c|c|c|c|c|c|c|c|c|c|c|c|c}
        $i$ & 1 & 2 & 3 & 4 & 5 & 6 & 7 & 8 & 9 & 10 & 11 & 12 & 13 & 14 \\\hline
        $T$ & $\texttt{a}$ & $\texttt{d}$ & $\texttt{a}$ & $\texttt{a}$ & $\texttt{a}$ & $\texttt{a}$ & $\texttt{b}$ & $\texttt{a}$ & $\texttt{a}$ & $\texttt{b}$ & $\texttt{b}$ & $\texttt{a}$ & $\texttt{a}$ & $\texttt{c}$
      \end{tabular}
  \end{center}
  For the dictionary $\D=\{\texttt{aa},\texttt{aaaa},\texttt{abba},\texttt{c}\}$,
    we have:
    \[\CD(5,12)=2,\
    \CD(2,6)=2,\
    \CD(2,12)=3.\]
    In particular, $T[5 \dd 12]$ contains two distinct patterns from $\D$: $\texttt{aa}$ (two occurrences) and $\texttt{abba}$.
  When the dictionary $\D$ represents all squares in $T$, we have
  \[\CD(5,12)=3,\ \CD(2,6)=2,\ \CD(2,12)=4. \]
  In particular, $T[5 \dd 12]$ contains three distinct squares: $\texttt{aa}$ (two occurrences), $\texttt{bb}$ and $\texttt{aabaab}$.
\end{example}
  
Let us note that one could answer $\CD(i,j)$ queries in time $\cO(j-i)$ by running $T[i \dd j]$ over the Aho--Corasick automaton of $\D$~\cite{DBLP:journals/cacm/AhoC75} or in time $\cOtilde(d)$ by performing Internal Pattern Matching~\cite{DBLP:conf/soda/KociumakaRRW15} for each element of $\D$ individually. 
Neither of these approaches is satisfactory as they can require $\Omega(n)$ time in the worst case.
  
\subparagraph{Our results and a roadmap.}
We start with preliminaries in~\cref{sec:prel} and an algorithmic toolbox in~\cref{sec:tools}.
Our results for the case of a static dictionary are summarized in Table~\ref{tab:results}.
Our solutions exploit string periodicity using runs and use data structures for variants of the (colored) orthogonal range counting problem and for auxiliary internal queries on strings.

\begin{table}[htpb!]
\begin{center}
    \begin{tabular}{c|c|c|c|c}
        \textbf{Space} & \textbf{Preprocessing time} & \textbf{Query time} & \textbf{Variant} & \textbf{Section} \\\hline
        $\Ohtilde(n+d)$ & $\Ohtilde(n+d)$ & $\Ohtilde(1)$ & 2-approximation & \ref{sec:approx} \\\hline
        $\Ohtilde(n^2/m^2+d)$ & $\Ohtilde(n^2/m+d)$ & $\Ohtilde(m)$ & exact & \ref{sec:trade1} \\\hline
         $\cOtilde(nd/m+d)$ & $\cOtilde(nd/m+d)$ & $\cOtilde(m)$ & exact & \ref{sec:trade2} \\\hline
        $\Oh(n \log^2 n)$ & $\Oh(n \log^2 n)$ & $\Oh(\log n)$ & $\D=$\,squares, exact & \ref{sec:squares}
    \end{tabular}
\end{center}
\caption{Our results for $\CD$ queries. Here, $m$ is an arbitrary parameter.}%
\label{tab:results}
\end{table}

For the case of a dynamic dictionary, where queries are interleaved with insertions and deletions of patterns in the dictionary, it was shown in~\cite{DBLP:conf/isaac/Charalampopoulos19} that the product of the time to process an update and the time to answer an $\Dec(i,j)$ query cannot be $\cO(n^{1-\epsilon})$ for any constant $\epsilon>0$, unless the Online Boolean Matrix-Vector Multiplication conjecture~\cite{DBLP:conf/stoc/HenzingerKNS15} is false.
In Section~\ref{sec:dynamic} we outline a general scheme that adapts our data structures for the case of a dynamic dictionary.
In particular, we show how to answer $\CD(i,j)$ queries $2$-approximately in $\cOtilde(m)$ time and process each update in $\cOtilde(n/m)$ time, for any $m$.

\section{Preliminaries}\label{sec:prel}

We begin with basic definitions and notation.
Let $T=T[1]T[2]\cdots T[n]$ be a \emph{string} of length $|T|=n$ over a linearly sortable alphabet $\Sigma$. The elements of $\Sigma$ are called \emph{letters}.
By $\varepsilon$ we denote an \emph{empty string}.
For two positions $i$ and $j$ on $T$, we denote by $T[i\dd j]=T[i]\cdots T[j]$ the \emph{fragment} of $T$ that starts at position $i$ and ends at position $j$ (the fragment is empty if $j<i$). A fragment is called \emph{proper} if $i>1$ or $j<n$. A fragment of $T$ is represented in $\cO(1)$ space by specifying the indices $i$ and $j$.
A \emph{prefix} of $T$ is a fragment that starts at position $1$
and a \emph{suffix} is a fragment that ends at position $n$.
By $UV$ and $U^k$ we denote the concatenation of strings $U$ and $V$ and $k$ copies of the string $U$, respectively.
A \emph{cyclic rotation} of a string $U$ is any string $V$ such that $U=XY$ and $V=YX$ for some strings $X$ and $Y$.

Let $U$ be a string of length $m$ with $0<m\leq n$. 
We say that $U$ is a \emph{factor} of $T$ if there exists a fragment $T[i \dd i+m-1]$, called an \emph{occurrence} of $U$ in $T$, that is matches $U$.
We then say that $U$ occurs at the \emph{starting position} $i$ in $T$.

A positive integer $p$ is called a \emph{period} of $T$ if $T[i] = T[i + p]$ for all $i = 1, \ldots, n - p$.  We refer to the smallest period as \emph{the period} of the string, and denote it by $\per(T)$. A string is called \emph{periodic} if its period is no more than half of its length and \emph{aperiodic} otherwise. The weak version of the periodicity lemma~\cite{fine1965uniqueness} states that if $p$ and $q$ are periods of a string $T$ and satisfy $p + q \leq |T|$, then $\gcd(p, q)$ is also a period of $T$.
A string $T$ is called \emph{primitive} if it cannot be expressed as $U^k$ for a string $U$ and an integer $k>1$.

The elements of the dictionary $\D$ are called \textit{patterns}. 
Henceforth, we assume that $\varepsilon \not\in \D$, i.e., that the length of each $P \in \D$ is at least $1$. 
We also assume that each pattern of $\D$ is given by the starting and ending positions of its occurrence in $T$. Thus, the size of the dictionary $d=|\D|$ refers to the number of patterns in $\D$ and not their total length.
A \emph{compact trie} of $\D$ is the trie of $\D$ in which all non-terminal nodes with exactly one child become implicit. The path-label $\Lab(v)$ of a node $v$ is defined as the path-ordered concatenation of the string-labels of the edges in the root-to-$v$ path.
We refer to $|\Lab(v)|$ as the \emph{string-depth} of $v$.

\section{Algorithmic Tools}\label{sec:tools}
\subsection{Modified Suffix Trees}
A \emph{$\D$-modified suffix tree}~\cite{DBLP:conf/isaac/Charalampopoulos19}, denoted as $\T_{T,\D}$, of a given text $T$ of length $n$ and a dictionary $\D$ is obtained from the trie of $\D \cup \{T[i\dd n] : 1\le i \le n\}$ by contracting, for each non-terminal node $u$ other than the root, the edge from $u$ to the parent of $u$. As a result, all the nodes of $\T_{T,\D}$ (except for the root) correspond to patterns in $\D$ or to suffixes of $T$. For $1\le i \le n$, the node representing $T[i\dd n]$ is labelled with $i$; see Figure~\ref{fig:dmod-example}.  For a dictionary $\D$ whose patterns are given as fragments of a text $T$, we can construct $\T_{T,\D}$ in $\cO(|\D|+|T|)$ time~\cite{DBLP:conf/isaac/Charalampopoulos19}.

\begin{figure}[h!]
  \centering
  \input{modifiedst.tex}
  \caption{Example of a $\D$-modified suffix tree
  for text $T=\texttt{adaaaabaabbaac}$ and dictionary $\D=\{\texttt{aa},\texttt{aaaa},\texttt{abba},\texttt{c}\}$ (figure from~\cite{DBLP:conf/isaac/Charalampopoulos19}).
 }%
  \label{fig:dmod-example}
\end{figure}
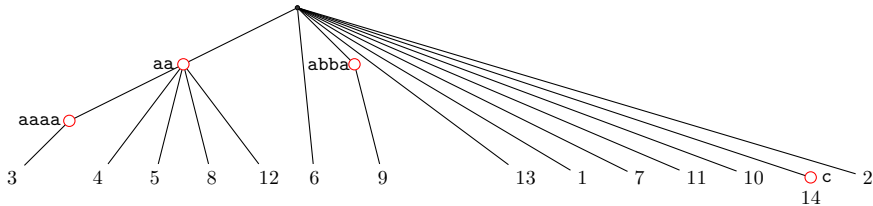
Let us denote by $\Occ(\D)$ the set of all occurrences of dictionary patterns in $T$, that is, the set of all fragments of $T$ that match a pattern in $\D$.
Using $\T_{T,\D}$, the set $\Occ(\D)$ can be computed in time $\Oh(n+d+|\Occ(\D)|)$.

We say that a tree is a \textit{weighted tree} if it is a rooted tree with an integer weight on each node $v$, denoted by $\omega(v)$, such that the weight of the root is zero and $\omega(u) < \omega(v)$ if $u$ is the parent of $v$.
We say that a node $v$ is a \textit{weighted ancestor at depth $\ell$} of a node $u$ if $v$ is the top-most ancestor of $u$ with weight of at least $\ell$.

\begin{theorem}[{\cite[Section 6.2.1]{DBLP:journals/talg/AmirLLS07}}]\label{thm:waq}
After $\cO(n)$-time preprocessing, weighted ancestor queries for nodes of a weighted tree $\mathcal{T}$ of size $n$ can be answered in $\cO(\log \log n)$ time per query.
\end{theorem}

The $\D$-modified suffix tree $\T_{T,\D}$ is a weighted tree with the weight of each node defined as the length of the corresponding string. We define the \emph{locus} of a fragment $T[i \dd j]$ in $\T_{T,\D}$ to be the weighted ancestor of the leaf $i$ at string-depth $j-i+1$.

\subsection{Auxiliary Internal Queries}

In a \emph{Bounded LCP} query, one is given two fragments $U$ and $V$ of $T$ and needs to return the longest prefix of $U$ that occurs in $V$; we denote such a query by $\BRLCP(U,V)$.
Kociumaka et al.~\cite{DBLP:conf/soda/KociumakaRRW15} presented several tradeoffs for this problem, including the following.

\begin{lemma}[\cite{DBLP:conf/soda/KociumakaRRW15},{\cite[Corollary 7.3.4]{tomeksthesis}}]\label{lem:BRLCP}
Given a text $T$ of length $n$, one can construct in $\cO( n \sqrt{\log n})$ time an $\cO(n)$-size data structure that answers Bounded LCP queries in $\cO(\log^{\epsilon}n)$ time, for any constant $\epsilon>0$.
\end{lemma}

Recall that $\C(i,j)$ returns the number of all occurrences of all the patterns of $\D$ in $T[i \dd j]$. The following result was proved in~\cite{DBLP:conf/isaac/Charalampopoulos19}.

\begin{lemma}[\cite{DBLP:conf/isaac/Charalampopoulos19}]\label{lem:count}
The $\C(i,j)$ queries can be answered in $\cO(\log^2 n/\log\log n)$ time with an $\cO(n+d \log n)$-size data structure, constructed in $\cO(n\log n / \log \log n +d\log^{3/2} n)$ time.
\end{lemma}

\subsection{Geometric Toolbox}
For a set of $n$ points in 2D, a range counting query returns the number of points in a given rectangle.
\begin{theorem}[Chan and Pătraşcu~\cite{DBLP:conf/soda/ChanP10}]\label{thm:range_count}
Range counting queries for $n$ integer points in 2D can be answered in time $\cO(\log n / \log \log n)$ with a data structure of size $\cO(n)$ that can be constructed in time $\cO(n \sqrt{\log n})$.
\end{theorem}

A quarterplane is a range of the form $(-\infty, x_1]\times (-\infty, x_2]$. By reversing coordinates we can also consider quarterplanes with some dimensions of the form $[x_i,\infty)$.
Let us state the following result on orthant color range counting due to Kaplan et al.~\cite{DBLP:journals/siamcomp/KaplanRSV08} in the special case of two dimensions.

\begin{theorem}[{\cite[Theorem 2.3]{DBLP:journals/siamcomp/KaplanRSV08}}]\label{thm:colored}
Given $n$ colored integer points in 2D, we can construct in $\cO(n\log n)$ time an $\cO(n\log n)$-size data structure that, given any quarterplane $Q$, counts the number of distinct colors with at least one point in $Q$ in $\cO(\log n)$ time.
\end{theorem}

We show how to apply geometric methods to a special variant of the $\CD$ 
problem, where we are interested in a small subset of occurrences of each pattern.

Let $\D=\{P_1,P_2, \ldots ,P_d\}$ and $\S$ be a family of sets $S_1,\ldots,S_d$ such that $S_k\subseteq \Occ(P_k)$, where $\Occ(P_k)$ is the set of positions of $T$ where $P_k$ occurs. Let $\|\S\|=\sum_k |S_k|$.
For each pattern $P_k$, we call the positions in the set $S_k$ the \emph{special positions of $P_k$}. 
Counting distinct patterns occurring at their special positions in $T[i \dd j]$ is called 
$\CD_{\S}(i,j)$.

\begin{lemma}\label{lem:fCD}
The $\CD_\S(i,j)$ queries can be answered in $\Oh(\log n)$ time with a data structure of size $\Oh(n+\|\S\| \log n)$ that can be constructed in $\Oh(n+\|\S\| \log n)$ time.
\end{lemma}
\begin{proof}
 We assign a different integer color $c_k$ to every pattern $P_k \in \D$.
 Then, for each fragment $T[a \dd b]= P_k$ such that $a \in S_k$, we add point $(a,b)$ with color $c_k$ in an initially empty 2D grid $\mathcal{G}$.
 A $\CD_\S(i,j)$ query reduces to counting different colors in the range $[i,\infty) \times (-\infty,j]$ of $\mathcal{G}$. The complexities follow from Theorem~\ref{thm:colored}.
\end{proof}

\subsection{Runs}
A \emph{run} (also known as a \emph{maximal repetition}) is a periodic fragment $R=T[a\dd b]$ which can be extended neither to the left nor to the right without increasing the period $p= \per(R)$, i.e.,
$T[a-1]\neq T[a +p-1] \text{ and } T[b-p+1] \neq T[b+1]$
provided that the respective positions exist.
If $\mathcal{R}$ is the set of all runs in a string $T$ of length $n$, then $|\mathcal{R}|\leq n$~\cite{DBLP:journals/siamcomp/BannaiIINTT17} and $\mathcal{R}$ can be computed in $\cO(n)$ time~\cite{KK:99}.
The \emph{exponent} $\exp(R)$ of a run $R$ with period $p$ is $|R|/p$. The sum of exponents of runs in a string of length $n$ is $\Oh(n)$~\cite{DBLP:journals/siamcomp/BannaiIINTT17,KK:99}.

The \emph{Lyndon root} of a periodic string $U$ is the  lexicographically smallest rotation of its $\per(U)$-length prefix. If $L$ is the Lyndon root of a periodic string $U$, then $U$ may be represented as  $(L ,r,a,b)$; here  $U= L[|L|-a+1\dd|L|]L^r L[1\dd b]$, and $r$ is called the \emph{rank} of $U$. Note that the minimal rotation of a fragment of a text can be computed in $\cO(1)$ time after an $\cO(n)$-time preprocessing~\cite{DBLP:conf/cpm/Kociumaka16}.

For a periodic fragment $U$, let $\run(U)$ be the run with the same period that contains $U$.  
\begin{lemma}[\cite{DBLP:journals/siamcomp/BannaiIINTT17,DBLP:journals/tcs/CrochemoreIKRRW14,tomeksthesis}]\label{lem:runS}
For a periodic fragment $U$, $\run(U)$ and its Lyndon root are 
uniquely determined and can be computed in constant time after linear-time preprocessing.
\end{lemma}

We use runs in 2-approximate $\CD(i,j)$ queries and in counting squares.

\section{Answering $\CD$ 2-Approximately}\label{sec:approx}

\subsection{CountDistinct for Extended or Contracted Fragments}

For two positions $\ell$ and $r$, we define $\Pref_{\D}(\ell,r)$ as the longest prefix of $T[\ell \dd r]$ that matches some pattern $P\in \D$; the length of such prefix is at most $r-\ell+1$. 
Let us show how to compute the locus of $\Pref_{\D}(\ell,r)$ in the $\D$-modified suffix tree $\T_{T,\D}$. 
To this end, we preprocess $\T_{T,\D}$ for weighted ancestor queries and store at every node $v$ of $\T_{T,\D}$ a pointer $p(v)$ to the nearest ancestor $u$ (including $v$) of $v$ such that $\Lab(u)\in \D$.
To return $\Pref_{\D}(\ell,r)$, we find the locus $u$ of $T[\ell\dd r]$ in the $\D$-modified suffix tree.
We return $p(u)$ if $|\Lab(u)|=|T[\ell\dd r]|$ and $p(v)$, where $v$ is the parent of $u$, otherwise.

\cref{lem:best} applies the $\D$-modified suffix tree to the problem of maintaining the count of distinct patterns occurring in a fragment subject to extending or shrinking the fragment.

\begin{lemma}\label{lem:best}
For any constant $\epsilon >0$, given $\CD(i,j)$, both $\CD(i\pm 1,j)$ and $\CD(i,j \pm 1)$ can be computed in $\Oh(\log^{\epsilon} n)$ time with an $\cO(n+d)$-size data structure that can be constructed in $\Oh(n\sqrt{\log n}+d)$ time. 
\end{lemma}
\begin{proof}
We only present a data structure that computes $\CD(i\pm 1,j)$ queries.
Queries $\CD(i,j \pm 1)$ can be handled analogously by building the same data structure for the reverses of all the strings in scope.

We show how to compute the number of patterns $P \in \D$ whose only occurrence in some fragment $T[\ell \dd r]$ starts at position $\ell$.
The computation of $\CD(i\pm 1,j)$ follows directly by setting $j=r$ and $\ell$ equal to $i-1$ or $i$.

\subparagraph{Data structure.} We preprocess $T$ for Bounded LCP queries (\cref{lem:BRLCP}) and construct the $\D$-modified suffix tree $\T_{T,\D}$ of text $T$ and dictionary $\D$.
In addition, we preprocess $\T_{T,\D}$ for weighted ancestor queries and store at every node $v$ of $\T_{T,\D}$ the number $\#(v)$ of the ancestors $u$ (including $v$) of $v$ such that $\Lab(u)\in \D$.

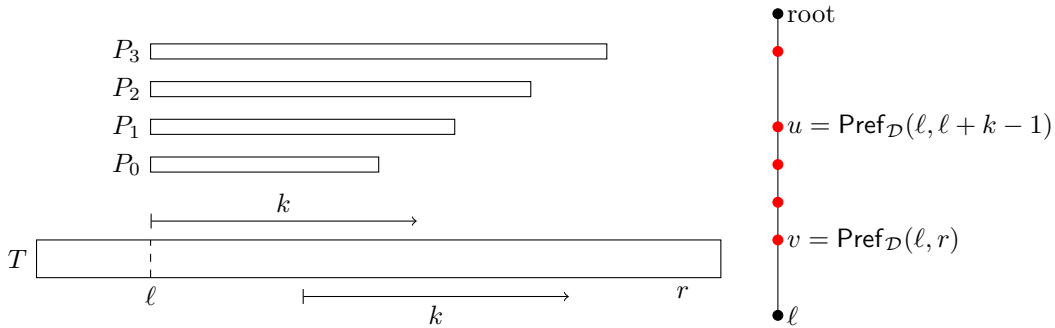
\begin{figure}[htpb!]
\begin{center}
\begin{tikzpicture}
\tikzset{
  dot/.style = {draw,fill,circle, minimum size=1.3mm, inner sep=0pt, outer sep=0pt}
}

\draw (0.5,0) rectangle (9.5,0.5);
\node at (0.5,0.25) [left] {$T$};
\node at (2, 0) [below] {$\ell$};
\node at (9, 0) [below] {$r$};
\draw [dashed] (2,0) -- +(0, 0.5);

\foreach \label/\y/\w in {$P_0$/1.5/3, $P_1$/2.0/4, $P_2$/2.5/5, $P_3$/3.0/6} {
  \node at (2,\y) [left] {\label};
  \draw (2,\y-0.1) rectangle +(\w,0.2);
}

\draw [|->] (2,0.75) -- +(3.5,0) node[midway,above] {$k$};

\draw [|->] (4,-0.25) -- +(3.5,0) node[midway,below] {$k$};

\begin{scope}[xshift=10.25cm, yshift=-0.5cm]
\node[dot] (root) at (0,4) {};
\node at (root) [right]  {root}; 
\node[dot] (l) at (0,0) {};

\node at (l) [right] {$\ell$}; 
\draw (root)--(l);
\node[dot,red] (u) at (0,2.5) {};
\node at (u) [right] {$u=\Pref_{\D}(\ell, \ell+k-1)$}; 
\node[dot,red] (v) at (0,1) {};
\node at (v) [right] {$v=\Pref_{\D}(\ell, r)$}; 

\node[dot,red] (p) at (0,1.5) {};
\node[dot,red] (q) at (0,2) {};
\node[dot,red] (z) at (0,3.5) {};
\end{scope}
\end{tikzpicture}
\end{center}
\caption{The setting of~\cref{lem:best}. Left: text $T$. Right: the path from the root of $\T_{T,\D}$ to the leaf with path-label $T[\ell \dd n]$.
The nodes of the path whose path-labels match some patterns from $\D$ are drawn in red.
Here, $P_0$ is the longest pattern that occurs at $\ell$ and also has an occurrence in $T[\ell+1 \dd r]$; its locus in $\T_{T,\D}$ is $u=\Pref_{\D}(\ell,\ell+k-1)$.
The patterns that occur in $T[\ell \dd r]$ only at position $\ell$ are $P_1, P_2$ and $P_3$. The locus of $P_3$ is $v=\Pref_{\D}(\ell,r)$.
Then, $\#(v)-\#(u)=5-2=3$.}\label{fig:best}
\end{figure}

\subparagraph{Query.} We want to count patterns 
longer than $k=|\BRLCP(T[\ell \dd r],T[\ell+1 \dd r])|$.
Let $u=\Pref_{\D}(\ell,\ell+k-1)$ and $v=\Pref_{\D}(\ell,r)$.
The desired number of patterns is equal to  $\#(v) -\#(u)$. See~\cref{fig:best} for a visualization.
\end{proof}

\subsection{Auxiliary Operation}

Two fragments $U=T[i_1\dd j_1]$ and $V=T[i_2\dd j_2]$ are called \emph{consecutive} if  $i_2 = j_1+1$. We denote the overlap $T[\max\{i_1,i_2\} \dd \min\{j_1,j_2\}]$ of $U$ and $V$ by $U \cap V$.

\defDSproblem{3-Fragments-Counting}{A text $T$ of length $n$ and a dictionary $\D$ consisting of $d$ patterns}{Given three consecutive fragments $F_1,F_2,F_3$ in $T$ such that $|F_1|=|F_3|$ and $|F_2| \ge 8 \cdot |F_1|$, count distinct patterns $P$ from $\D$ that have an occurrence starting in $F_1$ and ending in $F_3$ and do not occur in either $F_1F_2$ or $F_2F_3$}

Let us fix $|F_1| = |F_3| = x$ and $|F_2| = y \ge 8x$. Additionally, let us call an occurrence of $P \in \D$ that starts in fragment $F_a$ and ends in fragment $F_b$ an $(F_a, F_b)$-occurrence. We will call an $(F_1, F_3)$-occurrence an \emph{essential occurrence}.

We say that a string $S$ is \emph{highly periodic} if $\per(S) \le \frac14|S|$. We first consider the case that all patterns in $\D$ are not highly periodic.

\begin{lemma}\label{lem:3frag_aper}
  If each $P \in \D$ is not highly periodic, then 
\begin{multline*}
\textsc{3-Fragments-Counting}  (F_1,F_2,F_3) = \\ \C(F_1 F_2 F_3) - \C(F_1F_2) - \C(F_2F_3) + \C(F_2).
\end{multline*}
\end{lemma}

\begin{proof}
Let us start with the following claim.
\begin{claim}\label{uniqueocc}
Any $P \in \D$ that has an essential occurrence occurs exactly once in $F_1 F_2 F_3$.
\end{claim}
\begin{proof}
We have $|F_1 F_2 F_3| = x + y + x = 2x + y$.
String $P$ has an essential occurrence, so $|P| \geq y$.
Therefore, if there are two
occurrences of $P$ in $F_1 F_2 F_3$, then they overlap in
\[2|P| - (2x + y) \ge 2|P| - (\tfrac14|P| + |P|) = \tfrac34|P|\]
positions. This implies that $P$ is highly periodic, which is a contradiction.
\end{proof}

Claim~\ref{uniqueocc} shows that 
$\textsc{3-Fragments-Counting}(F_1,F_2,F_3)$ is equal to the
number of essential occurrences. Let us prove that the stated formula does not count any $(F_a, F_b)$-occurrences other than $(F_1, F_3)$-occurrences.

\begin{itemize}
\item Each $(F_1, F_2)$-occurrence is registered when we add $\C(F_1 F_2 F_3)$ and unregistered when we subtract $\C(F_1F_2)$. Similarly for $(F_2, F_3)$-occurrences.

\item Each $(F_2, F_2)$-occurrence is registered when we add $\C(F_1F_2F_3)$, $\C(F_2)$ and unregistered when we subtract $\C(F_1F_2)$, $\C(F_2F_3)$. 

\item Each $(F_1, F_1)$-occurrence is registered when we add $\C(F_1 F_2 F_3)$ and unregistered when we subtract $\C(F_1F_2)$. Similarly for $(F_3, F_3)$-occurrences.\qedhere
\end{itemize}
\end{proof}

We now proceed with answering \textsc{3-Fragments-Counting} queries for the dictionary of highly periodic patterns.

\begin{lemma}\label{lem:same_run}
 If $F_2$ is aperiodic, then there are no essential occurrences of highly periodic patterns.
 Otherwise, all essential occurrences of highly periodic patterns are generated by the same run, that is, $\run(F_2)$.
\end{lemma}
\begin{proof}
The first claim follows from the fact that such an occurrence of a pattern $P \in \D$ has an overlap of length at least $2\per(P)$ with $F_2$ and hence $\per(P)\leq \frac12|F_2|$ is a period of $F_2$.

As for the second claim, it suffices to show that, for any pattern $P \in \D$ that has an essential occurrence, we have $\per(P)=\per(F_2)$.
The inequalities $|F_2|\ge 2\per(F_2)$ and $|F_2| \ge 2\per(P)$ imply                                                $|F_2| \ge \per(F_2)+\per(P)$.
Hence, by the periodicity lemma, $q=\gcd(\per(P),\per(F_2))$ is a period of $F_2$.
As $q \leq \per(F_2)$, we conclude that $q=\per(F_2)$.
Thus, $\per(F_2)$ divides $\per(P)$, and therefore $\per(P)=\per(F_2)$. This concludes the proof.
\end{proof}

For a periodic factor $U$ of $T$, let \textsc{Periodic}$(U)$ denote the set of distinct patterns from $\D$ that occur in $U$ and have the same shortest period. 
Let us make the following observation.

\begin{observation}\label{obs:3fragments_per}
  If all $P \in \D$ are highly periodic, $F_2$ is periodic, and $R=\run(F_2)$, then 
\begin{multline*}
  \textsc{3-Fragments-Counting}(F_1,F_2,F_3) = \\ |\textsc{Periodic}(F_1 F_2 F_3 \cap R)| - |\textsc{Periodic}(F_1 F_2\cap R) \cup \textsc{Periodic}(F_2 F_3 \cap R)|.
\end{multline*}
\end{observation}

Next we now show how to efficiently evaluate the right-hand side of the formula in the observation above, using~\cref{thm:range_count} for efficiently answering range counting queries in 2D.

We group all highly periodic patterns by Lyndon root and rank; for a Lyndon root $L$ and a rank $r$, we denote by $\D^p_{L,r}$ the corresponding set of patterns.
Then, we build the data structure of~\cref{thm:range_count} for the set of points obtained by adding the point $(a,b)$ for each $(L,r,a,b) \in \D^p_{L,r}$.
We refer to the 2D grid underlying this data structure as $\G_{L,r}$.
Note that the total number of points in the data structures over all Lyndon roots and ranks is $\cO(d)$.

Each occurrence of a pattern $(L,r,a,b)$ lies within some run in $\mathcal{R}$ with Lyndon root $L$.
Let us state a simple fact.

\begin{fact}\label{fact:4cases}
    A periodic string $(L,r,a,b)$ occurs in a periodic string $(L,r',a',b')$ if and only if at least one of the following conditions is met:
    \begin{enumerate}[(1)]
     \item $r = r'$, $a \leq a'$, and $b \leq b'$;
     \item $r = r'-1$ and $a \leq a'$;
     \item  $r = r'-1$ and $b \leq b'$;
     \item $r \leq r'-2$.
    \end{enumerate}
\end{fact}

\begin{lemma}\label{lem:count_per}
One can compute $|\textsc{Periodic}(U)|$ for any periodic fragment $U$ in time $\cO(\log n / \log\log n)$ using a data structure of size $\cO(n+d)$ that can be constructed in time $\cO(n+d \sqrt{\log n})$.
\end{lemma}
\begin{proof}
For $U= (L,r,a,b)$, we count points contained in at least one of the rectangles
\begin{enumerate}[(1)]
    \item $(-\infty,a]\times (-\infty,b]$ in $\G_{L,r}$,
    \item $(-\infty,a]\times (-\infty,|L|]$ in $\G_{L,r-1}$,
    \item $(-\infty,|L|]\times (-\infty,b]$ in $\G_{L,r-1}$,
\end{enumerate}
and we add to the count the number of patterns of the form $(L,r',a,b)$ with $r'<r-1$. 
For the latter term, it suffices to store an array $X_L[1 \dd t]$ such that $X_L[r] = \sum_{i=1}^r |\D^p_{L,i}|$, where $t$ is the maximum rank of a pattern with Lyndon root $L$. The total size of these arrays is $\Oh(n)$ by the linearity of the sum of exponents of runs in a string~\cite{DBLP:journals/siamcomp/BannaiIINTT17,KK:99}.
\end{proof}

\begin{remark}\label{rem:tworanges}
In particular, in the proof of the above lemma, we count points that are contained within at least one out of a constant number of rectangles.
Therefore, not only we can easily compute $|\textsc{Periodic}(U)|$, but similarly we are able to compute $|\textsc{Periodic}(U_1) \cup \textsc{Periodic}(U_2)|$ for some periodic factors $U_1, U_2$ of $T$.
\end{remark}

We are now ready to prove the main result of this subsection.

\begin{lemma}\label{lem:3count}
The \textsc{3-Fragments-Counting}$(F_1,F_2,F_3)$ queries can be answered in $\cO(\log^2 n / \log\log n)$ time with a data structure of size $\cO(n+d \log n)$ that can be constructed in $\cO(n\log n / \log \log n + d \log^{3/2} n)$ time.
\end{lemma}
\begin{proof}
By~\cref{lem:3frag_aper}, in order to count the patterns that are not highly periodic, it suffices to perform three $\C$ queries. To this end, we employ the data structure of~\cref{lem:count} which answers $\C$ queries in $\cO(\log^2 n/\log\log n)$ time, occupies space $\cO(n+d \log n)$, and is constructed in time $\cO(n\log n / \log \log n +d\log^{3/2} n)$.

We now proceed to counting highly periodic patterns.
First, we check whether $F_2$ is periodic; this can be done in $\cO(1)$ time after an $\cO(n)$-time preprocessing of $T$~\cite{tomeksthesis,DBLP:conf/soda/KociumakaRRW15}.
If $F_2$ is not periodic, then by~\cref{lem:same_run} no highly periodic pattern has an essential occurrence, and we are thus done.
If $F_2$ is periodic, three $|\textsc{Periodic}(U)|$ queries suffice to obtain the answer due to \cref{obs:3fragments_per}. They can be efficiently answered due to~\cref{lem:count_per,rem:tworanges}; the complexities are dominated by those for building the data structure for $\C$ queries.
\end{proof}

\subsection{Approximation Algorithm}

Let us fix $\eps=\frac19$.
A fragment of length $\floor{(1+\eps)^p}$ for any positive integer $p$ will be called a \emph{$p$-basic fragment}. Our data structure stores $\CD(i,j)$ for every basic fragment $T[i \dd j]$. Using Lemma~\ref{lem:best}, these values can be computed in $\Oh(n\log^{1+\epsilon} n+d)$ time with a sliding window approach. The space requirement is $\cO(n\log n +d)$.

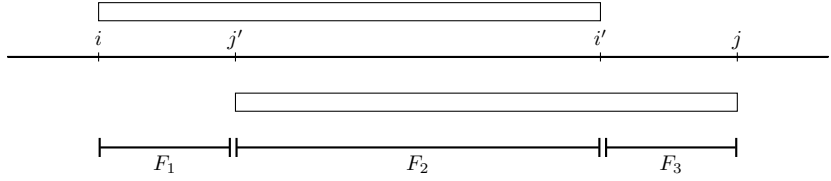
\begin{figure}[h!]
  \centering
  \input{threeranges.tex}
  \caption{A 2-approximation of $\CD(i,j)$  is achieved using precomputed counts for basic factors   $T[i \dd i']$ and $T[j' \dd j]$.
 }\label{fig:threeranges}
\end{figure}

In order to answer an arbitrary $\CD(i,j)$ query, let $T[i \dd i']$ and $T[j' \dd j]$ be the longest prefix and suffix of $T[i\dd j]$ being a basic factor; see Figure~\ref{fig:threeranges}. We sum up $\CD(i,i')$ and $\CD(j',j)$ and the result of a \textsc{3-Fragments-Counting} query for $F_1=T[i \dd j'-1]$, $F_2=T[j' \dd i']$, $F_3=T[i'+1\dd j]$. (Note that $(|F_1| + |F_2|) \cdot (1 + \eps) > |F_1| + |F_2| + |F_3|$ implies $\eps (|F_1| + |F_2|) > |F_3|$, and since $|F_1|=|F_3|$, we have that $|F_1|= |F_3| \le \frac18 |F_2|$.) 
Now, a pattern $P \in \D$ is counted at least once if and only if it occurs in $T[i \dd j]$.
Also, a pattern $P \in \D$ is counted at most twice (exactly twice if and only if it occurs in both $F_1F_2$ and $F_2F_3$). The above discussion and~\cref{lem:3count} yield the following result.

\begin{theorem}\label{thm:approx}
$\CD(i,j)$ queries can be answered 2-approximately in $\cO(\log^2 n / \log\log n)$ time with a data structure of size $\cO((n+d) \log n)$ that can be constructed in time $\cO(n\log^{1+\epsilon} n + d \log^{3/2} n)$ for any constant $\epsilon >0$.
\end{theorem}

\section{Time-Space Tradeoffs for Exact Counting}

\subsection{Tradeoff for Large Dictionaries}\label{sec:trade1}

The following result is yet another application of~\cref{lem:best}.

\begin{theorem}\label{thm:ext}
  For any $m \in [1,n]$ and any constant $\epsilon >0$, the $\CD(i,j)$ queries can be answered in $\Oh(m \log^{\epsilon} n)$ time using an $\Oh(n^2/m^2+n+d)$-size data structure that can be constructed in $\Oh((n^2 \log^{\epsilon} n)/m + n\sqrt{\log n}+d)$ time.
\end{theorem}
\begin{proof}
A fragment of the form $T[c_1 m+1 \dd c_2 m]$ for integers $c_1$ and $c_2$ will be called a \emph{canonical fragment}.
Our data structure stores $\CD(i',j')$ for every canonical fragment $T[i' \dd j']$ and the data structure of~\cref{lem:best}. Hence the space complexity $\Oh(n^2/m^2+n+d)$. 

We can compute in $\cO(n\log^{\epsilon} n)$ time $\CD(i',j)$ for a given $i'$ and all $j$ using Lemma~\ref{lem:best}.
There are $\cO(n/m)$ starting positions of canonical fragments and hence the counts for all canonical fragments can be computed in $\Oh((n^2\log^{\epsilon} n)/m)$ time. Additional preprocessing time $\Oh(n\sqrt{\log n}+d)$ originates from \cref{lem:best}.

\begin{figure}[h!]
\begin{center}
\begin{tikzpicture}
  \tikzset{
    dot/.style = {draw,fill,circle, minimum size=1mm, inner sep=0pt, outer sep=0pt}
  }
  \draw (0,0)--(12,0);

  \foreach \x in {0, 2, 4, 6, 8, 10, 12} {
    \draw (\x,-0.1)--(\x,0.1);
  }

  \node[dot] (i) at (1.2,0) {};
  \node at (i) [below] {$\vphantom{i'j'}i$};
  \node[dot] (j) at (7.5,0) {};
  \node at (j) [below] {$\vphantom{i'j'}j$};
  
  \draw (2,0) node[below] {$\vphantom{i'j'}i'$};
  \draw (6,0) node[below] {$\vphantom{i'j'}j'$};

  \draw (2,0.5) rectangle (6, 0.7);
  \draw (1.2,0.5) rectangle (2, 0.7);
  \draw (6,0.5) rectangle (7.5, 0.7);

  \draw[snake=brace] (2, 0.8) -- (6, 0.8) node [above, midway] {\small canonical fragment} ; 

  \draw[|-latex] (2, 1.0) -- (1.2, 1.0) node [above, midway] {\small extend};
  \draw[|-latex] (6, 1.0) -- (7.5, 1.0) node [above, midway] {\small extend};
\end{tikzpicture}
\end{center}
\caption{An illustration of the setting in the query algorithm underlying~\cref{thm:ext}.}\label{fig:canon}
\end{figure}
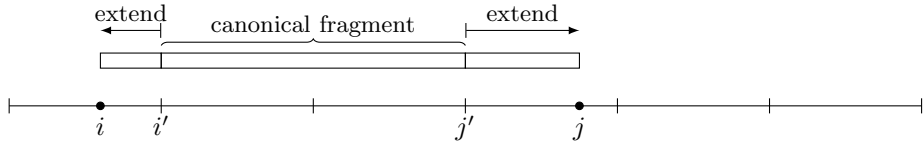

We can answer a $\CD(i,j)$ query in $\cO(m \log^{\epsilon} n)$ time as follows.
Let $T[i'\dd j']$ be the maximal canonical fragment contained in $T[i\dd j]$.
We retrieve $\CD(i',j')$ for $T[i' \dd j']$.
Then, we apply~\cref{lem:best} $\cO(m)$ times; each time we extend the fragment for which we count, until we obtain $\CD(i,j)$. See Figure~\ref{fig:canon}.
\end{proof}

\subsection{Tradeoff for Small Dictionaries}\label{sec:trade2}

We call a set of strings $\PS$ a \emph{path-set} if all elements of $\PS$ are prefixes of its longest element.
We now show how to efficiently handle dictionaries that do not contain large path-sets.

\begin{lemma}\label{lem:smallsets}
If $\D$ does not contain any path-set of size greater than $k$, then we can construct in $\cO(kn \log n)$ time an $\cO(kn \log n)$-size data structure that can answer $\CD(i,j)$ queries in $\cO(\log n)$ time.
\end{lemma}
\begin{proof}
Let $\D=\{P_1,\ldots,P_d\}$ and $\S=\{\Occ(P_1),\ldots,\Occ(P_d)\}$. Every position of $T$ contains at most $k$ occurrences of patterns from $\D$.
This implies that $\|\S\| \le kn$.
A $\CD(i,j)$ query can obviously be treated as a $\CD_\S(i,j)$ query.
The complexities follow from Lemma~\ref{lem:fCD}.
\end{proof}

\begin{lemma}\label{lem:path_decomp}
 For any $k\in [1,n]$, we can compute a maximal family $\mathcal{F}$ of pairwise-disjoint path-sets in $\D$, each consisting of at least $k$ elements, in $\cO(n+d)$ time.
\end{lemma}
\begin{proof}
Let us consider the $\D$-modified suffix tree $\T_{T,\D}$ and call every its terminal node that has no descendant terminal nodes a \emph{bottom node}.
As the considered path-sets are maximal, the longest string in any path-set $\PS \in \mathcal{F}$ is a bottom node.
We preprocess $\T_{T,\D}$ so that for each bottom node $u \in \T_{T,\D}$ we store a counter $C(u)$ equal to the number of terminal nodes on the root-to-$u$ path.

We perform a preorder traversal of $\T_{T,\D}$.
This way all bottom nodes in $\T_{T,\D}$ are considered in a left-to-right manner.
When adding a path-set to $\mathcal{F}$, we mark all nodes of that path-set.
During our traversal we can easily maintain the number $N$ of ancestors of the node that we are visiting that have been marked.
When we visit some bottom node $u$, we check whether $r=C(u)-N$ is at least $k$.
In such case we add the path-set consisting of $u$ and its unmarked ancestors being terminal nodes to $\mathcal{F}$.
Note that throughout the above process we maintain that if a terminal node is marked, then all its ancestor terminal nodes are also marked. Hence we can easily find the $r$ unmarked terminal nodes that are ancestors of $u$ since they are $u$'s $r$ closest ancestors being terminal -- we can store for each terminal node a pointer to its closest ancestor that is terminal.
\end{proof}

We now combine Lemmas~\ref{lem:smallsets}, \ref{lem:path_decomp} and \ref{lem:BRLCP} to get the main result of this section.

\begin{theorem}\label{main1}
For any $m\in [1,n]$ and any constant $\epsilon>0$, the $\CD(i,j)$ queries can be answered in  $\cO(m \log^{\epsilon} n +\log n)$ time using an $\cO((nd\log n) / m +d)$-size data structure that can be constructed in $\cO((nd \log n) / m +d)$ time.
\end{theorem}
\begin{proof}
We first apply~\cref{lem:path_decomp} for $k=\ceil{d/m}$.
We then have a decomposition of $\D$ to a family $\mathcal{F}$ of at most $m$ path-sets and a set $\D'$ with no path-set of size greater than $\floor{d/m}$.
We directly apply~\cref{lem:smallsets} for $\D'$. In order to handle path-sets, we build the data structure of~\cref{lem:BRLCP}.
Then, upon a $\CD(i,j)$ query, for each path-set $\PS \in \mathcal{F}$, we compute the longest pattern in $\PS$ that occurs in $T[i \dd j]$ using a Bounded LCP query followed by a predecessor query~\cite{DBLP:journals/ipl/Willard83} in a structure that stores the lengths of the elements of $\PS$, with the lexicographic rank in $\PS$ stored as satellite information.
The data structure of~\cite{DBLP:journals/ipl/Willard83} is randomized,
but it can be combined with deterministic dictionaries~\cite{DBLP:conf/icalp/Ruzic08} using a simple two-level 
approach (see~\cite{DBLP:conf/stoc/Thorup03a}), resulting in a deterministic \emph{static} data structure.
\end{proof}

\begin{remark}
Let us fix the query time to be $\cO(m \log^{\epsilon }n)$ for $m=\Omega(\log n)$. 
Then,~\cref{main1} outperforms~\cref{thm:ext} in terms of the required space for $d=o(n/(m\log n))$.
For example, for $m=d=n^{1/4}$, the data structure of~\cref{main1} requires space $\cOtilde(n)$ while the one of~\cref{thm:ext} requires space $\cOtilde(n\sqrt{n})$.
\end{remark}

\section{Internal Counting of Distinct Squares}\label{sec:squares}
The number of occurrences of squares could be quadratic,
but we can construct a smaller $\Oh(n \log n)$-size subset of these occurrences 
(called \emph{boundary occurrences}) that, from the point 
of view of $\CD$ queries, gives almost the same answers.
This is the main trick 
in this section.
Distinct squares with a boundary occurrence in a given fragment can 
be counted in $\Oh(\log n)$ time due to Lemma~\ref{lem:fCD}.
The remaining squares can be counted based on their structure: we show that they are all generated by the same run.

Now, the dictionary $\D$ is the set of all squares in $T$. By the following fact, $d = \Oh(n)$ and $\D$ can be computed in $\Oh(n)$ time.

\begin{fact}[\cite{DBLP:journals/tcs/CrochemoreIKRRW14,DBLP:journals/dam/DezaFT15,DBLP:journals/jct/FraenkelS98,DBLP:journals/jcss/GusfieldS04}]
A string $T$ of length $n$ contains $\Oh(n)$ distinct square factors and they can all be computed in $\Oh(n)$ time.
\end{fact}

We say that an occurrence of a square $U^2$ is \emph{induced} by a run $R$ if it is contained in $R$ and the shortest periods of $U$ and $R$ are the same. Every occurrence of a square is induced by exactly one run.

We need the following fact (note that it is false for the set of \emph{all} runs; see~\cite{DBLP:journals/tcs/GlenS13}).

\begin{fact}\label{sum}
The sum of the lengths of all highly periodic runs is $\cO(n \log n)$.
\end{fact}
\begin{proof}
We will prove that each position in $T$ is contained in $\cO(\log n)$ highly periodic runs.
Let us consider all highly periodic runs $R$ containing some position $i$, such that $m \leq \per(R)< \frac32 m$ for some even integer $m$.
Suppose for the sake of contradiction that there are at least $5$ such runs.
Note that each such run fully contains one of the fragments $T[i-3m+1+t \dd i+t]$ for $t \in \{0,m,2m,3m\}$.
By the pigeonhole principle, one of these four fragments is contained in at least two runs, say $R_1$ and $R_2$.
In particular, the overlap of these runs is at least $3 m \geq \per(R_1)+\per(R_2)$, which is a contradiction by the periodicity lemma.
 \end{proof}
We define a family of occurrences $\B=B_1,\ldots,B_d$ such that, for each square $U_i^2$, the set $B_i$ contains the leftmost and the rightmost occurrence of $U_i^2$ in every run. We call these \emph{boundary occurrences}. Boundary occurrences of squares have the following property.
 
\begin{lemma}\label{boundary}
$\|\B\|=\Oh(n \log n)$ and the set family $\B$ can be computed in $\Oh(n \log n)$ time.
\end{lemma}

\begin{proof}
Let us define the \emph{root} of a square $U^2$ to be $U$.
A square is primitively rooted if its root is a primitive string.
Let \emph{p-squares} be primitively rooted squares, \emph{np-squares} be the remaining ones. 
The number of occurrences of p-squares in a string of length $n$ is $\Oh(n \log n)$ and they can all be computed in $\Oh(n \log n)$ time; see~\cite{DBLP:journals/ipl/Crochemore81,DBLP:journals/tcs/StoyeG02}.

We now proceed to np-squares.
 Note that for any highly periodic run $R$, the leftmost occurrence of each np-square induced by $R$ starts in one of the first $\per(R)$ positions of $R$; a symmetric property holds for rightmost occurrences and last $\per(R)$ positions.
In addition, it can be readily verified that such a position is the starting (resp.~ending) position of at most $\exp(R)$ squares induced by $R$.
It thus suffices to bound the sum of $\exp(R)\cdot \per(R)$ over all highly periodic runs $R$.
The fact that $\exp(R)\cdot \per(R)=|R|$ concludes the proof of the combinatorial part by 
Fact~\ref{sum}.

For the algorithmic part, it suffices to iterate over the $\cO(n)$ runs of $T$.
\end{proof}

\begin{lemma}\label{lem:aperc}
If $T[i \dd j]$ is non-periodic, $\CD(i,j) = \CD_\B(i,j)$.
\end{lemma}
\begin{proof}
Let us consider an occurrence of a square $U^2$ inside $T[i \dd j]$. Let $R$ be the run that induces this occurrence. By the assumption of the lemma, $R$ does not contain $T[i \dd j]$. Then at least one of the boundary occurrences of $U^2$ in $R$ is contained in $T[i \dd j]$.
\end{proof}

For a periodic fragment $F$ of $T$, by $\RunSquares(F)$ we denote the number of distinct squares that are induced by $F$ (being a run if interpreted as a standalone string). 
The value $\RunSquares(F)$ can be computed in $\Oh(1)$ time, as it was shown in e.g.~\cite{DBLP:journals/tcs/CrochemoreIKRRW14}.

Let $F_1$ be a prefix and $F_2$ be a suffix of a periodic fragment $F$, 
such that each of $F_1$ and $F_2$ is of length at most $\per(F)$ -- and hence they are disjoint. By $\Sql(F,F_1,F_2)$ (``bounded squares'') we denote the number of distinct squares induced by $F$ which have an occurrence starting in $F_1$ or ending in $F_2$.

\begin{lemma}\label{lem:SQL}
Given $\per(F)$, the $\Sql(F,F_1,F_2)$ queries can be answered in $\Oh(1)$ time.
\end{lemma}

\begin{proof}
We are to count distinct squares induced by $F$ that start in $F_1$ or end in $F_2$.

We introduce an easier version of $\Sql$ queries. Let $\Sql'(F,F_1)=\Sql(F,F_1,\varepsilon)$ 
be the number 
of squares induced by $F$ which start in its prefix $F_1$ of length at most $p:= \per(F)$.

\subparagraph{Reduction of $\Sql$ to  $\Sql'$.}
First, observe that the set of squares induced by $F$ starting at some position $q\in [1,p]$ and the set of squares induced by $F$ ending at some position $q'\in [|F|-p+1,|F|]$ are equal if $q \equiv q'+1 \pmod{p}$ and disjoint otherwise.
Also note that $F_2=UV$ for some prefix $V$ and some suffix $U$ of $F[p]F[1 \dd p-1]$; we consider this rotation of $F[1 \dd p]$ to offset the $+1$ factor in the above modular equation.
Let $|U|=a$ and $|V|=b$.

Then, by the aforementioned observation, we are to count distinct squares that start in some position in the set $[1, |F_1|] \cup [1,b] \cup [p-a+1,p]$; see Figure~\ref{fig:lem27}.

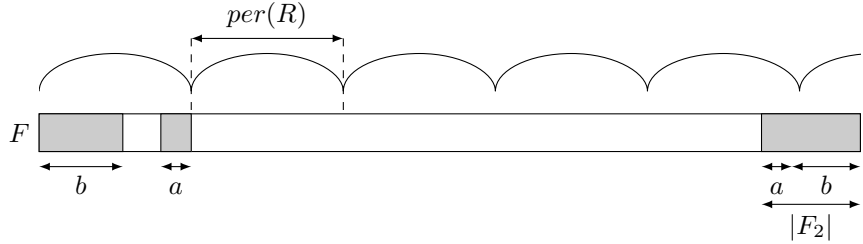
\begin{figure}[h!]
    \centering
    \begin{tikzpicture}
    
      \draw (0,0) node [left, yshift=0.25cm] {$F$} rectangle (10.8,0.5) ;
    
      \draw [fill=white!80!black] (0, 0) 
        rectangle (1.1,0.5);
      \draw [fill=white!80!black] (1.6, 0) rectangle (2.0,0.5);
      \draw [fill=white!80!black] (9.5, 0) rectangle (10.8,0.5);
    
      \draw [latex-latex] (0,-0.2)--(1.1,-0.2) node[midway,below] {$\vphantom{ab}b$}; 
      \draw [latex-latex] (1.6,-0.2)--(2.0,-0.2) node[midway, below] {$\vphantom{ab}a$}; 
      \draw [latex-latex] (9.5,-0.2)--(9.9,-0.2) node[midway, below] {$\vphantom{ab}a$}; 
      \draw [latex-latex] (9.9,-0.2)--(10.8,-0.2) node[midway, below] {$\vphantom{ab}b$}; 
      \draw [latex-latex] (9.5,-0.7)--(10.8,-0.7) node[midway, below] {$|F_2|$}; 
    
      \begin{scope}
      \clip (0.0,0) rectangle (10.8,2);
        \draw[snake=bumps,segment length=4cm, segment amplitude=0.5cm, line after snake=0pt, segment aspect=0] 
        (0.0, 0.8) -- +(15.0, 0); 
      \end{scope}
    
      \draw[latex-latex] (2.0, 1.5) -- +(2.0,0) node [midway, above] {$per(R)$};
      \draw[dashed] (2.0, 1.5)--+(0, -1.0);
      \draw[dashed] (4.0, 1.5)--+(0, -1.0);

    \end{tikzpicture}
    \caption{
    Reduction of $\Sql$ to  $\Sql'$; the case that $|F_1|\leq b$.
    }\label{fig:lem27}
\end{figure}

Hence the computation of $\Sql(F,F_1,F_2)$ is reduced to at most two instances of 
the special case when $F_2$ is the empty string.

\subparagraph{Computation of $\Sql'(F,F_1)$.}
The number of squares induced by $F$ starting at $F[i]$ is equal to $\lfloor{(|F|-i+1)/(2p)}\rfloor$.
Consequently,
$\Sql'(F,F_1)=\sum_{i=1}^{|F_1|}\lfloor{(|F|-i+1)/(2p)}\rfloor=  |F_1|\cdot t - \max\{0, |F_1|-k-1\}$,
where $t=\lfloor |F|/(2p) \rfloor$ and $k=|F| \bmod (2p)$.
\end{proof}

\begin{lemma}\label{lem:perc}
Assume that $F=T[i \dd j]$ is periodic and $R=T[a \dd b]=\run(T[i \dd j])$. Let $F_1=T[i \dd a+p-1]$ and $F_2 = T[b-p+1 \dd j]$, where $\per(R)=p$. Then:
\begin{equation}\label{formula}
\CD(i,j) = \CD_\B(i,j) + \RunSquares(F) - \Sql(F,F_1,F_2).
\end{equation}
\end{lemma}
\begin{proof}
In the sum $\CD_\B(i,j) + \RunSquares(F)$, all squares are counted once except for squares whose boundary occurrences are induced by $R$, which are counted twice. They are exactly counted in the term $\Sql(F,F_1,F_2)$; see~\cref{fig:tworanges}.
\end{proof}

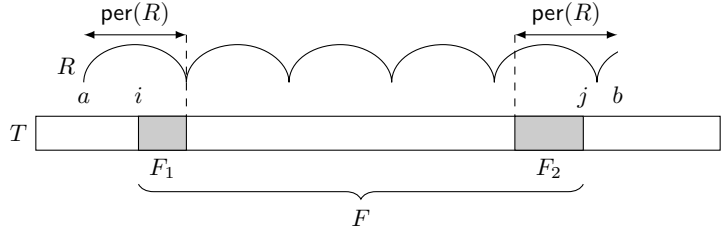
\begin{figure}[h!]
\begin{center}
\begin{tikzpicture}[scale=0.9, transform shape]

\draw (0,0) node [left, yshift=0.25cm] {$T$} rectangle (10,0.5) ;

\draw [fill=white!80!black] (1.5, 0) node[above, yshift=0.5cm] {$\vphantom{aijb}i$} 
  rectangle (2.2,0.5) node [midway, below, yshift=-0.25cm] {$F_1$};

\draw [fill=white!80!black] (7, 0) rectangle (8,0.5) 
  node[above] {$\vphantom{aijb}j$}
  node [midway, below, yshift=-0.25cm] {$F_2$};

\draw[snake=brace, segment amplitude=0.2cm] (8,-0.5)--(1.5,-0.5);
\draw (4.75,-0.75) node[below] {$F$};

\begin{scope}
\clip (0.7,0) rectangle (8.5,2);
  \draw[snake=bumps,segment length=2.7cm, segment amplitude=0.5cm, line after snake=0pt, segment aspect=0] 
  (0.7, 1) -- +(12, 0); 
\end{scope}

  \node [left] at (0.7,1.25) {$R$}; 
  
  \node [above] at (0.7,0.5) {$\vphantom{aijb}a$}; 
  \node [above] at (8.5,0.5) {$\vphantom{aijb}b$}; 

  \draw[latex-latex] (0.7, 1.7) -- +(1.5,0) node [midway, above] {$\per(R)$};
  \draw[dashed] (2.2, 1.7)--+(0, -1.2);
  \draw[latex-latex] (7, 1.7) -- +(1.5,0) node [midway, above] {$\per(R)$};
  \draw[dashed] (7, 1.7)--+(0, -1.2);
\end{tikzpicture}
\end{center}

\caption{The setting in~\cref{lem:perc}. Note that $F_1$ is empty if $i\ge a+\per(R)$; similarly for $F_2$.}\label{fig:tworanges}
\end{figure}

\begin{theorem}
If $\D$ is the set of all square factors of $T$, then $\CD(i,j)$ queries can be answered in $\Oh(\log n)$ time using a data structure of size $\Oh(n \log^2 n)$ that can be constructed in $\Oh(n \log^2 n)$ time.
\end{theorem}
\begin{proof}
We precompute the set $\B$ in $\Oh(n \log n)$ time using Lemma~\ref{boundary} and perform $\Oh(n \log^2 n)$ time and space preprocessing for $\CD_\B(i,j)$ queries. 

In order to answer a $\CD(i,j)$ query, first we ask a $\run(T[i\dd j])$ query of Lemma~\ref{lem:runS} to check if $T[i \dd j]$ is periodic. 

We compute $\CD_\B(i,j)$ which takes $\Oh(\log n)$ time
due to Lemma~\ref{lem:fCD}. If $T[i \dd j]$ is non-periodic, then it is the final result
due to Lemma~\ref{lem:aperc}.

Otherwise $T[i \dd j]$ is periodic. Let $F,F_1,F_2$ be as in Lemma~\ref{lem:perc}.
We answer $\RunSquares(F)$ and $\Sql(F,F_1,F_2)$ queries in $\Oh(1)$ time using the algorithm from~\cite{DBLP:journals/tcs/CrochemoreIKRRW14} and Lemma~\ref{lem:SQL}, respectively. 
Finally, $\CD(i,j)$ is computed using~\eqref{formula}. 
\end{proof}

\section{Dynamic Dictionary}\label{sec:dynamic}

The general framework for dynamic dictionaries essentially consists in rebuilding a static data structure after every $k$ updates.
We return correct answers by performing individual queries for the patterns inserted or deleted from the dictionary since the data structure was built.
In particular, we show that an application of this framework --with some tweaks-- to the data structure of~\cref{sec:approx} yields the following.

\begin{theorem}\label{thm:dyn_approx}
For any $k\in [1, n]$, we can construct a data structure in $\cOtilde(n+d)$ time, which processes each update to the dictionary in $\cOtilde(n/k)$ time and answers $\CD(i,j)$ queries 2-approximately in $\cOtilde(k)$ time.
\end{theorem}

Let $\Dec(P,i,j)$ denote the query checking whether some pattern $P$ occurs in $T[i \dd j]$. We make use of the following result.

\begin{theorem}[\cite{DBLP:journals/tcs/KellerKFL14}]\label{thm:ipm}
$\Dec(P,i,j)$ queries can be answered in time $\cO(\log\log n)$ with an $\cO(n \log^{\epsilon} n)$-size data structure that can be constructed in $\cO(n \sqrt{\log n})$ time.
\end{theorem}

\begin{remark}
Actually, in~\cite{DBLP:journals/tcs/KellerKFL14} there is an extra $|P|+\occ$ additive factor in the query time complexity as the pattern need not be given as a fragment of $T$ and the authors want to output all occurrences in $T[i \dd j]$. The $|P|$ factor corresponds to computing the locus of the pattern in the suffix tree of $T$, which we can do instead using~\cref{thm:waq}.
\end{remark}

\subparagraph{General scheme.}
This general scheme is analogous to what we used in order to dynamize data structures for the other internal dictionary matching queries in~\cite{DBLP:conf/isaac/Charalampopoulos19}.
Let us suppose that we can build in $p(n,d)$ time a data structure that answers $\CD(i,j)$ queries (exactly) in $q(n,d)$ time.
We construct this data structure over the input text $T$ and dictionary $\D=\D_0$, where $d=|\D|$.
We also preprocess the text for internal pattern matching queries, as per~\cref{thm:ipm}.
Then, for the first $k$ updates to the dictionary we answer $\CD(i,j)$ queries in $\cOtilde(q(n,d)+k)$ time by asking a $\CD(i,j)$ for $\D_0$ and then querying for each inserted/deleted pattern individually, using internal pattern matching queries.
After $k$ updates, we update our data structure to refer to dictionary $\D_k$ in time $u(n,d,k)$ -- thus, each update gets processed in $\cOtilde(u(n,d,k)/k)$ amortized time\footnote{E.g.~one can rebuild the data structure from scratch in $\cO(p(n,d+k))$ time.}.
The time complexity can be deamortized by employing the standard time slicing technique.
Then, if we want queries and updates to cost roughly equal we need to balance $q(n,d)+k=u(n,d,k)/k$.

\subparagraph{Dynamic 2-approximation.} We now apply this general scheme to our data structure for computing a $2$-approximation of $\CD(i,j)$. First of all, on each query, we need to check whether each pattern that has been deleted from $\D_0$ is counted once or twice by the static data structure for $\D_0$. For this, it suffices to query whether such pattern occurs in the two relevant basic factors.

We update our data structure after $k$ updates to the dictionary as follows.
\begin{itemize}
\item Let $\D_{del}=\D_0 \setminus \D_k$ and $\D_{ins}=\D_k \setminus \D_0$. We adjust $\CD(i,j)$ for each basic factor in $\cOtilde(n+k)$ time by counting distinct patterns of $\D_{del}$ and $\D_{ins}$ in each of them, as in the preprocessing of~\cref{thm:approx}.
\item We maintain our collections of points on grids $\mathcal{G}_{L,r}$ using the dynamic data structure of Chan and Tsakalidis for 2D range counting, which supports updates and queries in $\cOtilde(1)$ time~\cite{DBLP:conf/compgeom/ChanT17}. As for the values $\sum_{i=1}^r |\D^p_{L,i}|$, we use an augmented balanced binary search tree.
\item Finally, we can update the data structure for $\C(i,j)$ upon a batch of $k=\cO(n)$ updates to the dictionary in time $\cOtilde(n)$ as shown in~\cite{DBLP:conf/isaac/Charalampopoulos19}.
\end{itemize}

This concludes the proof of \cref{thm:dyn_approx}.


\bibliographystyle{plainurl}
\bibliography{references}
\end{document}

%% file: modifiedst.tex
\begin{tikzpicture}[scale=0.75, transform shape]

  \tikzstyle{v}=[draw, circle, minimum width=0.75pt, inner sep=0.75pt, fill=black!80!white];
  \tikzstyle{marked}=[draw, shape=circle, minimum width=1pt, inner sep=2pt, red];
  \tikzstyle{label}=[midway,inner sep=1pt,above, sloped, fill=none];
  \tikzstyle{label2}=[midway,inner sep=1pt];

  \node[v] (r) at (5, 8) {};

  \foreach \d/\label in {0.3/6, 4/13, 5/1, 6/7, 7/11, 8/10, 10/2} {
    \node (r\label) at ($(r.center)+(\d, -3)$) {\label};
    \draw (r) -- (r\label);
  }

  \node[marked] (abba) at ($(r.center)+(1, -1)$) {};
  \draw (r) -- (abba);
  \node[left of=abba, node distance=13pt] {$\texttt{abba}$};
  \node (abba2) at ($(abba.center)+(0.5, -2)$) {9};
  \draw (abba) -- (abba2);

  \node[marked] (c) at ($(r.center)+(9, -3)$) {};
  \draw (r) -- (c);
  \node [below of=c, node distance=10pt] {14};
  \node[right of=c, node distance=8pt] {$\texttt{c}$};

  \node[marked] (aa) at ($(r.center)+(-2, -1)$) {};
  \node[left of=aa, node distance=10pt] {$\texttt{aa}$};

  \foreach \d/\label in {-1.5/4,-0.5/5, 0.5/8, 1.5/12} {
    \node (aa\label) at ($(aa.center)+(\d, -2)$) {\label};
    \draw (aa) -- (aa\label);
  }

  \node[marked] (aaaa) at ($(aa.center)+(-2, -1)$) {};
  \node[left of=aaaa, node distance=15pt] {$\texttt{aaaa}$};
  \node (aaaa2) at ($(aaaa.center)+(-1, -1)$) {3};

  \draw (r) -- (aa);
  \draw (aa) -- (aaaa);
  \draw (aaaa) -- (aaaa2);

\end{tikzpicture}

%% file: threeranges.tex
\begin{tikzpicture}[scale=0.6, every node/.style={scale=0.8}]
\draw [fill, thick] (0,0) rectangle (18,0);

\foreach \x/\t in {2/i, 5/j', 13/i', 16/j/}	
  	\draw [black] (\x,-0.1) -- (\x, 0.1) node [above, midway] (\t) {$\vphantom{i'j'}\t$};


\draw (5,-0.8) rectangle (16,-1.2);


\draw (2,0.8) rectangle (13,1.2);


\foreach \x/\y/\t in {2/4.9/F_1, 5/13/F_2, 13.1/16/F_3}	
{
  	\draw [thick, |-|] (\x,-2.0) --  (\y,-2.0) node [below, midway] {$\t$} ; 
};
\end{tikzpicture}